\documentclass{amsart}
\usepackage{graphicx}
\usepackage{bbm}
\usepackage{color}
\newtheorem{theorem}{Theorem}[section]
\newtheorem{lemma}[theorem]{Lemma}

\theoremstyle{definition}

\newtheorem{example}[theorem]{Example}

\theoremstyle{remark}

\numberwithin{equation}{section}

\newtheorem{rem}[theorem]{Remark}

\newcommand{\Fa}{\mathcal{F}}
\newcommand{\pr}{\mathbb P}
\newcommand{\E}{\mathbb E}
\newcommand{\R}{\mathbb R}
\newcommand{\wt}{\widetilde}
\allowdisplaybreaks

\begin{document}

\title[ON MVH UNDER PARTIAL OBSERVATIONS  AND TERMINAL WEALTH CONSTRAINTS]{ON MEAN-VARIANCE HEDGING UNDER PARTIAL OBSERVATIONS  AND TERMINAL WEALTH CONSTRAINTS}

\author{VITALII MAKOGIN}

\address[Vitalii Makogin]{Institute of Stochastics,\\  Ulm University, D-89069, Ulm, Germany\\
vitalii.makogin@uni-ulm.de} 

\author{ALEXANDER MELNIKOV}

\address[Alexander Melnikov]{Department of Mathematical and Statistical Sciences,\\  University of Alberta, 632 Central Academic Building\\
Edmonton, AB T6G 2G1., Canada\\
melnikov@ualberta.ca}

\author{YULIYA MISHURA}

\address[Yuliya Mishura]{Department of Probability Theory, Statistics and Actuarial
Mathematics,\\ Taras Shevchenko National University of Kyiv, Volodymyrska 64,\\
Kyiv, 01601, Ukraine\\ myus@univ.kiev.ua }

\maketitle

\begin{abstract}
In the paper a mean-square minimization problem under terminal wealth constraint with partial observations is studied. The problem is naturally connected to the mean-variance hedging problem under incomplete information. A new approach to solving this problem is proposed. The paper provides a solution when the underlying pricing process is a square-integrable semimartingale. The proposed method for study is based on the martingale representation. In special cases the Clark-Ocone representation can be used to obtain explicit solutions. The results and the method are illustrated and supported by example with two correlated geometric Brownian motions.
\end{abstract}


\section{Introduction}

Let us start with the problem which can arise on commodity markets. Particularly, the following example is given in the paper of \cite{Weisshaupt}. Suppose that a mining company wants to exploit a certain mining area. The mineral $A$ which the mining company would produce after this investment is not traded now, but there exists another mineral $\widetilde{A}$ with similar characteristics  which is already traded on the market. Then the unobservable price process $S$ of the mineral $A$ is highly correlated with the price process $\widetilde{S}$ of the mineral $\widetilde{A}.$ The mining company wants to insure against the risk that the price of the mineral $A$ after the period of investment $T$ will be  below the production costs $c.$ It invests money in options on mineral $\wt{A}.$ In this case, the company wants to find a strategy that minimizes the expected value of the squared difference between the price of $A$ and portfolio value. Such problem is called the problem of mean-variance hedging (MVH).

Since the pioneering work of \cite{folmer}, the  mean-variance hedging
is a permanent area of research in mathematical finance. At the beginning, the problem was formulated
assuming that the probability measure was a martingale measure. In this context, some results were obtained in the case of full information by  \cite{folmer}.
Under incomplete information this type of hedging  was  developed later in many papers. See, for example the paper of
\cite{Schweizer_92}, where two correlated Wiener processes were considered and the hedging strategy could depend on
both, and \cite{Schweizer_94} where results were obtained using projection techniques. 

We introduce the model of the present paper as follows.
Assume that an agent used to work on well-known simple financial market and this market is observable enough to be complete. The agent wants to use financial instruments only from this market, and he wants to start working on a new unobservable market. We assume that this new market is incomplete. In our model there are two contingent claims $\wt{H}$ and $H$. The first one is observable and from the completeness of the observable market the agent superhedges $\wt{H}$,  {i.e., he chooses the strategy under which the terminal wealth exceeds the claim value.} The second claim $H$ is from the new incomplete market and for reducing the new risk the agent uses the mean-variance approach, {i.e., he simultaneously minimizes the mean-square difference of the terminal wealth and the unobservable claim. The simplest example of this model is the mean-variance hedging of unobservable claim $H$ using observable asset with the condition that portfolio value is non-negative at time $T,$ i.e., in this case $\wt{H}=0.$}
{We assume that there are no arbitrage on both markets and therefore the value process of an admissible portfolio with non-negative terminal value is also non-negative. It enables us to consider the terminal constraint rather than the constraint which is pointwise in time.}

To authors' knowledge such an approach has not been considered in the existing literature. {In the paper of \cite{folmer2000} the expectation of shortfall weighted by some loss function is  minimized under condition that wealth process is non-negative on interval $[0,T].$}
{In the paper of \cite{Korn95} investors trade in such a way that they achieve a non-negative wealth over the whole time interval $[0, T].$ The authors' analysis is based on the expected utility or the mean-variance approach and they gave a common framework including both types of selection criteria as special cases by considering portfolio problems with terminal wealth constraints. In the paper of \cite{Korn97} the considered problem is the mean-variance minimization for Black-Sholes model under condition that the terminal wealth only is non-negative. But author's aim  was concentrated on finding the terminal wealth value and the exact form of the minimization strategy was not given.}
{In the paper of \cite{Bielecki} a continuous-time mean-variance portfolio selection problem is studied where all the market coefficients are random and the wealth process under any admissible trading strategy is not allowed to be below zero at any time.}

In the  paper \cite{Fujii}, the mean-variance hedging problem is studied in a partially observable market where the drift processes can only be inferred through the observation of asset or index processes. In the paper of \cite{Ceci} authors 
provide a suitable Galtchouk-Kunita-Watanabe decomposition of the contingent claim that works in a partial information framework. The paper \cite{Hubalek} considers  variance-optimal hedging for processes with stationary independent increments.  And in the paper of \cite{Jeanblanc} the problem of mean-variance hedging
for general semimartingale models is solved via stochastic control methods.

Mathematically, the ``unobservable" information is described by the filtration $\mathbb{F}=\{\Fa_t, t\geq 0\}$ and ``observable" information is described by its subfiltration $\wt{\mathbb{F}}=\{\widetilde{\Fa}_t, t\geq 0\}.$ We restrict ourselves to the finite horizon $T>0.$ The prices of underlying asset are  given by $\wt{\mathbb{F}}$-adapted square-integrable stochastic process $\wt{S}=\{\wt{S}_t,t\in [0,T]\}.$  
We assume that $\wt{S}$ be a semimartingale. 
A trading strategy is described by $\wt{\mathbb{F}}$-predictable square-integrable stochastic processes $\xi=\{\xi_t,t\in [0,T]\}$ such that the stochastic integral $\int_0^T \xi_t d\wt{S}_t$ is well-defined. This integral describes the trading gains induced by the self-financing portfolio strategy associated to $\xi.$ Let a contingent claim $\wt{H}$ be a $\wt{\Fa}_T$-measurable square-integrable nonnegative random variable. At time $T$, a hedger who starts with initial capital $x$ and uses the strategy $\xi,$ has to pay the random amount $\wt{H},$ so that portfolio value should not be less than $\wt{H}.$ This contingent claim $\wt{H}$ can be interpreted as a random lower bound on a terminal wealth. At the same time hedger wants to approximate a random amount $H$ by portfolio value. In contradiction to $\wt{H}$, we assume that $H$ is a $\Fa_T$-measurable square-integrable non-negative random variable. 
In this context, mean-variance hedging  problem with partial observations  means solving the optimization problem 
$$ \mbox{minimize } \E \left(H-x-\int_0^T \xi_t d\wt{S}_t \right)^2 \mbox{ over all } \xi \in \Xi (x,\wt{H}),$$
where
$\Xi (x,\wt{H})=\{\xi: x+\int_0^T \xi_t d\wt{S}_t \ge \wt{H} ~a.s.\}$

This problem is naturally related to the mean-variance hedging. The main challenge in solving mean-variance hedging problem is to find more explicit descriptions of the optimal strategy {and in the present paper we introduce the new approach.}
{
The concrete form of the minimizing strategy can be obtained by the Clark-Ocone formula when it is applicable.}


We consider the example of this problem when pricing processes of observable and unobservable assets $\wt{S}$, $S$ are two correlated geometric Wiener processes. {In this case we can use the delta-hedging but we illustrate how our method works by applying the Clark-Ocone theorem in the Brownian setting.} 
For the case when the contingent claim $H$ is a call-option we provide the precise formula of the solution and make numerical illustrations.


The paper is organized as follows. In section 2 we formulate the conditional mean-variance hedging problem under incomplete information in the general semimartingale setting. In section 3 we reduce it to the simplified statement in order to avoid technical details. We prove the auxiliary result concerning the representation of the random variable that is approached and prove the main result that gives the solution of the minimization problem.  In  section 4 the corresponding results are illustrated with the help of the model with two correlated geometric Wiener processes. In subsection 4.1 we present the numerical illustrations.  

\section{Preliminaries and the Formulation of a Problem}

Let we have complete probability space $(\Omega, \Fa, P)$  with filtration $\mathbb{F}=\{\Fa_t, t\geq 0\}$ that corresponds to the ``complete information". Suppose that there exists a subfiltration  $\mathbb{\widetilde{F}}=\{\widetilde{\Fa}_t, t\geq 0\}$ that corresponds to the ``incomplete information". Consider {c\`adl\`ag} risk asset $\widetilde{S}=\{\widetilde{S_t},t\geq 0\}$ such that  $\widetilde{S}$ is non-negative and adapted to the  ``incomplete information'', or, that is the same for us,  $\mathbb{\widetilde{F}}$-adapted. We suppose that non-risky asset $B_t \equiv 1$ and the market $\Sigma$ is arbitrage-free on $(\Omega, \Fa, P)$ with filtration $\mathbb{F}$. Moreover, we suppose that $\widetilde{S}\in\Sigma$ and the observable market $\widetilde{\Sigma}={\{1,\widetilde{S}\}}$ is complete on  $(\Omega, \Fa, P)$  with filtration $\mathbb{\widetilde{F}}\subset \mathbb{F}$. We restrict ourselves to the finite horizon $T$, so we consider all   processes on the interval $[0,T]$.
Let $\mathcal{P}$ be the set of all equivalent martingale measures for $\Sigma$. Then the restriction $\widetilde{P} $ of any  $P^*\in \mathcal{P}$ on $\mathbb{\widetilde{F}}$ is the same unique equivalent martingale measure for the observable market $\widetilde{\Sigma}$ so that $\widetilde{S}$ is $\mathbb{\widetilde{F}}$-martingale w.r.t. $ \widetilde{P}.$  Denote $D_T=\frac{d \wt{P}_T}{d P_T}$ the restriction of $\frac{d \wt{P}}{d P}$ on interval $[0,T]$. 

 Now we introduce two square-integrable nonnegative random variables, $H$ and $\widetilde{H}$, $H$ being $ {\Fa}_T$-measurable and $\widetilde{H}$ being $\widetilde{\Fa}_T$-measurable. We can characterize them as ``unobservable'' and ``observable'' random variables or contingent claims, correspondingly. 

In order to remain within the framework of the square-integrable approach, we introduce the following assumptions.
\begin{enumerate}
\label{prob_A}
\item[$(A1)$] $\widetilde{S}=\{\widetilde{S_t},t\geq 0\}$ is the semimartingale admitting the representation $\widetilde{S_t}=\widetilde{N_t}+\widetilde{A_t},$ where $\widetilde{N}$ is the square-integrable martingale and $\widetilde{A}$ is the  predictable process of square-integrable variation. Suppose that $\mathbb{\widetilde{F}}$ is generated by $\widetilde{N}=\{\widetilde{N}_t,t\geq 0\}$.
\item[$(A2)$] $\E D_T^2<\infty,$ $\E H^2<\infty $ and $\E \widetilde{H}^2<\infty.$
\end{enumerate}
These conditions mean, in particular, that we can consider stochastic integral w.r.t. the semimartingale $\widetilde{S},$ $$ {I}(t, {\xi})=\int_0^t \xi_s d \wt{S}_s=\int_0^t \xi_s d \wt{N}_s+\int_0^t \xi_s d \wt{A}_s,\;t\in[0,T]$$ for such $\mathbb{\widetilde{F}}$-predictable processes $\xi$ that  $  \int_0^T \xi_s^2 d \langle \widetilde{N}\rangle_s <\infty$ and $ \int_0^T |\xi_s| d |\widetilde{A}|_s <\infty $ a.s. 

Denote $\Xi $ class of such $\mathbb{\widetilde{F}}$-predictable strategies.

Completeness of market $\wt{\Sigma}$ together with condition $(A2)$ means that for any initial value $x\geq \E_{\widetilde{P} }\widetilde{H}$ we can construct the superhedge of the contingent claim $\widetilde{H}$ a.s. with the help of   such    $\xi\in\Xi $    that
 $\E( {I}(T, {\xi}))^2<\infty$. In other words, there exists such $\xi\in\Xi $ that
 $x+{I}(T, {\xi})\geq \widetilde{H}  \quad a.s.$ Further, for any $x \in \R$ denote 
$$\Xi(x,\widetilde{H})=\{\xi\in \Xi : x+\int_0^T \xi_s  d   \widetilde{S}_s\geq \widetilde{H}\;\textit{a.s.}\}.$$
Note, if $x<\E_{\wt{P}} \wt{H},$ then the space of admissible strategies is empty. So, in what follows, we can assume that $x\geq \E_{\wt{P}} \wt{H}.$

Now we can state a conditional minimization problem in the semimartingale framework.

\textbf{Problem} $[S({x,{\widetilde{H}};{H}})]$. Starting with fixed value $x\geq \E_{\widetilde{P} }\widetilde{H},$ to construct the hedging strategy $\widetilde{\xi}\in\Xi(x, \widetilde{H})$ so that
$$\E\left(H-x-\int_0^T \widetilde{\xi} _s d \wt{S}_s \right)^2=\min_{\xi \in \Xi(x, \widetilde{H})}\E\left(H-x-\int_0^T \xi_s d \wt{S}_s\right)^2, $$
and such $\widetilde{\xi}$ for which $$\min_{\xi \in \Xi(x,\widetilde{H})}\E\left({H}-x-\int_0^T \xi_s d \wt{S}_s\right)^2=\E\left({H}-x-\int_0^T \widetilde{\xi}_s d \wt{S}_s\right)^2.$$

Note that if $\widetilde{N}$ has a martingale representation property then the the market $\widetilde{\Sigma}$ is complete. For example, Brownian motion and compensated Poisson process have this property.

\section{Main results}

In order to simplify the solving  of Problem  $[S({x,{\widetilde{H}};{H}})]$ we make the following remarks.
\begin{rem}\label{rem6}{\rm
 We present $\E\left(H-x-\int_0^T \xi_s d \wt{S}_s\right)^2$ as
\begin{align*}
&\E\left(H-\E(H\vert\wt{\Fa}_T)+\E(H\vert\wt{\Fa}_T)-x-\int_0^T \xi_s d \wt{S}_s\right)^2\\
&=\E\left(H-\E(H\vert\wt{\Fa}_T)\right)^2+\E\left(\E(H\vert\wt{\Fa}_T)-x-\int_0^T \xi_s d \wt{S}_s\right)^2,
\end{align*} 
since $H-\E(H\vert\wt{\Fa}_T)$ and any square-integrable $\wt{\Fa}_T$-measurable random variable  are orthogonal. So, Problem  $[S({x,{\widetilde{H}};{H}})]$ is reduced to the finding of
\begin{equation}
\label{pr_G}
\begin{gathered}
\min_{\xi \in \Xi(x,\widetilde{H})}\E\left(\E(H|\widetilde{\Fa}_T)-x-\int_0^T \xi_s d \wt{S}_s\right)^2,
 \end{gathered}
\end{equation}
and such $\widetilde{\xi}$ for which $$\min_{\xi \in \Xi(x,\widetilde{H})}\E\left(\E(H|\widetilde{\Fa}_T)-x-\int_0^T \xi_s d \wt{S}_s\right)^2=\E\left(\E(H|\widetilde{\Fa}_T)-x-\int_0^T \widetilde{\xi}_s d \wt{S}_s\right)^2.$$
}
\end{rem}
\begin{rem}\label{rem7} {\rm 
Denote $H_1=\E(H|\widetilde{\Fa}_T).$ Consider the expansion  $H_1=H_1\mathbbm{1}_{H_1\geq \widetilde{H}}+H_1\mathbbm{1}_{H_1<\widetilde{H}}$
and denote $H_2=H_1\mathbbm{1}_{H_1\geq \widetilde{H}}+\widetilde{H}\mathbbm{1}_{H_1<\widetilde{H}}\geq \widetilde{H }$. 
On one hand, for any $\xi \in \Xi(x,\widetilde{H})$
\begin{align}
\nonumber
&\E\left(H_1-x-\int_0^T \xi_s d \wt{S}_s\right)^2  \\
\nonumber & =\E\left(H_1-x-\int_0^T \xi_s d \wt{S}_s\right)^2\mathbbm{1}_{H_1<\widetilde{H}} 
+\E\left(H_1-x-\int_0^T \xi_s d \wt{S}_s\right)^2\mathbbm{1}_{H_1\geq\widetilde{H}} \\
\label{equ--1} &\geq \E\left(\widetilde{H}-x-\int_0^T \xi_s d \wt{S}_s\right)^2\mathbbm{1}_{H_1<\widetilde{H}} +\E\left(H_1-x-\int_0^T \xi_s d \wt{S}_s\right)^2\mathbbm{1}_{H_1\geq\widetilde{H}}\\
\nonumber &=\E\left(H_2-x-\int_0^T \xi_s d \wt{S}_s\right)^2.
\end{align}

On the other hand,  for any $\xi \in \Xi(x,\widetilde{H})$
\begin{align}
\nonumber&\E\left(H_1-x-\int_0^T \xi_s d \wt{S}_s\right)^2=\E\left(H_1\mathbbm{1}_{H_1\geq \widetilde{H}}+H_1\mathbbm{1}_{H_1<\nonumber\widetilde{H}}-x-\int_0^T \xi_s d \wt{S}_s\right)^2 \\
\nonumber&=\E\left(H_1\mathbbm{1}_{H_1\geq \widetilde{H}}-x-\int_0^T \xi_s d \wt{S}_s\right)^2+\E\left(H_1\mathbbm{1}_{H_1<\widetilde{H}}\right)^2\\
\nonumber& -2\E\left(x+\int_0^T \xi_s d \wt{S}_s\right)H_1\mathbbm{1}_{H_1<\widetilde{H}}\\
\nonumber&\leq \E\left(H_1\mathbbm{1}_{H_1\geq \widetilde{H}}-x-\int_0^T \xi_s d \wt{S}_s\right)^2+\E\left(H_1\mathbbm{1}_{H_1<\widetilde{H}}\right)^2-2\E\left(\widetilde{H}H_1\mathbbm{1}_{H_1<\widetilde{H}}\right)\\
&\label{equ--2} \leq \E\left(H_1\mathbbm{1}_{H_1\geq \widetilde{H}}-x-\int_0^T \xi_s d \wt{S}_s\right)^2-\E\left(H_1\mathbbm{1}_{H_1<\widetilde{H}}\right)^2,
\end{align}
and the equalities in \eqref{equ--1},\eqref{equ--2} are achieved if and only if $H_1\geq\widetilde{H}$ a.s. Therefore, we can  restrict ourselves to the case    $H_1\geq\widetilde{H}$ a.s. and in other cases apply bounds \eqref{equ--1},\eqref{equ--2} (It  will be specified in Remark \ref{rem10} how to deal with  the  upper bound in \eqref{equ--2}.)}
\end{rem}
\begin{rem}{\rm \label{rem8} Now, let $H_1\geq\widetilde{H}$ a.s. and consider the case when $x=\E_{\widetilde{P}} H_1 .$  It follows from the completeness of the market $\wt{\Sigma}$ that we   have  the representation $H_1=x+\int_0^T \xi^0_s d \wt{S}_s$ for some $\xi^0 \in \Xi(x, \wt{H})$. So, we put $\widetilde{\xi}=\xi^0$ and get the trivial zero solution of minimization problem. So, it is reasonable to  consider two cases: $x<\E_{\widetilde{P}} H_1 $ and $x>\E_{\widetilde{P}} H_1 $. However, since our goal is to solve the minimization problem with minimal initial resources, we suppose in what follows that $x<\E_{\widetilde{P}} H_1 $.}
\end{rem}
\begin{rem}\label{rem9} { \rm Further, let $\E_{\widetilde{P}}\widetilde{H}  \leq x<\E_{\widetilde{P}} H_1 $. 
Evidently, $\E\left(H_1-x-\int_0^T \xi_s d \wt{S}_s\right)^2$ $=\E\left(H_1-\wt{H}-\left(x+\int_0^T \xi_s d \wt{S}_s-\wt{H}\right)\right)^2.$ It follows from the completeness of the market  that there exists ${\eta} \in \Xi$ such that $\wt{H}=\E_{\widetilde{P}}\wt{H}+\int_0^T {\eta}_s d \wt{S}_s.$   Now, rewrite 
\begin{equation}
\nonumber
\E\left(H_1-x-\int_0^T \xi_s d \wt{S}_s\right)^2=\E\left(H_1-\widetilde{H}-\left(x-\E_{\widetilde{P}}\wt{H}+\int_0^T (\xi_s-\eta_s) d \wt{S}_s\right)\right)^2,
\end{equation}
denote $G=H_1-\widetilde{H}$, $g=x-\widetilde{x}$, $ {\zeta}_s=\xi_s-\eta_s,$ and note that $G\geq 0$ a.s., {$\E G>0$}.
Then obviously $0\leq g<\E_{\widetilde{P} }G$ and $g+\int_0^T \zeta_s d \wt{S}_s=x-\widetilde{x}+\int_0^T (\xi_s-\eta_s) d \wt{S}_s\geq \widetilde{H}-\widetilde{H}= 0$ a.s.  The case $g=0$  corresponds to $x=\E_{\widetilde{P}}\wt{H}$ and then the trivial solution of $[S({x,{\widetilde{H}};{H}})]$ is the strategy ${\eta}_s.$
 So, further we assume $g>0$ or $x>\E_{\widetilde{P} }\wt{H}.$
It means that we reduce Problem $[S({x,{\widetilde{H}};{H}})]$ to the following one.}
\end{rem}

\textbf{Problem} $[S({g,0;{G}})]$. For fixed square-integrable nonnegative $ {\widetilde{F}}_T$-measurable random variable $G$ and fixed number  $0< g< \E_{\widetilde{P} }G$  to find  $$\min_{\xi \in \Xi(g,0)}\E\left(G-g-\int_0^T \xi_s d \wt{S}_s\right)^2,$$ and such $\widetilde{\xi}\in \Xi(g,0)$ for which $$\min_{\xi \in \Xi(g,0)}\E\left(G-g-\int_0^T \xi_s d \wt{S}_s\right)^2=\E\left(G-g-\int_0^T \widetilde{\xi}_s d \wt{S}_s\right)^2.$$
\begin{rem}\label{rem10} { \rm Consider the term $\E(x, \wt{S}):=\E\left(H_1\mathbbm{1}_{H_1\geq \widetilde{H}}-x-\int_0^T \xi_s d \wt{S}_s\right)^2$ from \eqref{equ--2}. We can present it as
\begin{equation*}
\E(x, \wt{S})= \E\left((H_1-\widetilde{H})\mathbbm{1}_{H_1\geq \widetilde{H}}-\left(x+\int_0^T \xi_s d \wt{S}_s-\widetilde{H}\mathbbm{1}_{H_1\geq \widetilde{H}}\right)\right)^2.
\end{equation*}
It follows from the market completeness that $\widetilde{H}\mathbbm{1}_{H_1\geq \widetilde{H}}$ admits  the representation $$\widetilde{H}\mathbbm{1}_{H_1\geq \widetilde{H}}=\widetilde{x}_1+\int_0^T \widetilde{\gamma}_s d \wt{S}_s,$$
where $\widetilde{x}_1=\E_{\widetilde{P}}\widetilde{H}\mathbbm{1}_{H_1\geq \widetilde{H}}\leq \E_{\widetilde{P}}\widetilde{H}\leq x$ and $x+\int_0^T \xi_s d \wt{S}_s-\widetilde{H}\mathbbm{1}_{H_1\geq \widetilde{H}} \geq 0$. Therefore, the minimization of the term $\E(x, \wt{S})$ is in the framework of the Problem $[S({g,0;{G}})]$ with $g=x-\widetilde{x}_1\geq 0$ and $G=(H_1-\widetilde{H})\mathbbm{1}_{H_1\geq \widetilde{H}}\geq 0$. So, in the general case, when the inequality $H_1\geq\widetilde{H}$ does not hold a.s.,
we can minimize right-hand side of \eqref{equ--2}  in the framework of the Problem $[S({g,0;{G}})]$ and find the minimization strategy $\wt{\xi}_t.$ The minimum value  $\min_{\xi \in \Xi(x,\wt{H})} \E\left(\E(H|\widetilde{\Fa}_T)-x-\int_0^T \xi_s d \wt{S}_s\right)^2$ will be between the right-hand side of \eqref{equ--1} evaluated for strategy $\wt{\xi}_t+\wt{\gamma}_t$
 and  the minimum value of the right-hand side of \eqref{equ--2}.

}
\end{rem}
To solve Problem $[S({g,0;{G}})]$, recall   $D_T=\frac{d \wt{P}_T}{d P_T}$ is the restriction of $\frac{d \wt{P}}{d P}$ on $[0,T]$  and note that {for any $v(x)\in \R$} $$\E_{\widetilde{P}} \left(G+v(x)D_T\right)^{+} \leq  \E D_TG+ |v (x)|\E D^2_T<+\infty.$$ Now, for any $0< x< \E_{\widetilde{P} }G$ consider $v(x)$ that is the solution of equation
 \begin{equation}\label{MMM-1-2}
\E_{\widetilde{P}}\left(G+v(x)D_T\right)^+=x.
\end{equation}
We see that $v(x)$ is non decreasing.
\begin{lemma}
\label{v_eq}
Function $v=v(x),\; x\in(0,\E_{\widetilde{P} }G)$ is uniquely determined, continuous and {strictly} increasing on the interval $(0, \E_{\widetilde{P}}  G )$. The range of values of this function is the interval $(r, 0)$, where $r=-{\mathrm{ess}\sup}_{\omega\in \Omega}G/D_T.$
\end{lemma}

\begin{proof} Evidently,   function $f(y)=\E_{\widetilde{P}}\left(yD_T+G\right)^+$ is non-decreasing and continuous on  $\mathbb{R}$, $\lim_{y\rightarrow-\infty}f(y)=0$ and $\lim_{y\rightarrow+\infty}f(y)=+\infty$. Therefore, for any $x\in(0,\E_{\widetilde{P}} G)$ solution of equation \eqref{MMM-1-2} exists. Furthermore, let for $y_1<y_2$
\begin{equation}\label{MMM-1-1}
\E_{\widetilde{P}}\left(y_1 D_T+G\right)^+=\E_{\widetilde{P}}\left(y_2 D_T+G\right)^+=x>0.
\end{equation}
Then $\widetilde{P}\{G>-y_1 D_T\}>0$ therefore with positive probability $\left(y_2 D_T+G\right)^+>$ $\left(y_1 D_T+G\right)^+$ that is in contradiction with \eqref{MMM-1-1}. Therefore, solution of equation \eqref{MMM-1-2}
 is unique and function $v=v(x), v\in(0,\E_{\widetilde{P}} G)$ is uniquely determined. 
 Now, $\E_{\widetilde{P}}\left(v(x)D_T+G\right)^+$ {strictly} increases in $x\in(0,\E_{\widetilde{P}} G)$. Indeed, let $\E_{\widetilde{P}}G >x_2>x_1>0$ and $$\E_{\widetilde{P}}\left(v(x_2) D_T+G\right)^+=x_2>\E_{\widetilde{P}}\left(v(x_1) D_T+G\right)^+=x_1>0.$$ It means in particular that $$\mathbbm{1}\left\{v(x_2) D_T+G>0\right\}- \mathbbm{1}\left\{v(x_1) D_T+G>0\right\}>0$$ with positive probability. Therefore
  \begin{equation*}
  \left(v(x_2) D_T+G\right)\mathbbm{1}\left\{v(x_2) D_T+G>0\right\}>\left(v(x_1) D_T+G\right)\mathbbm{1}\left\{v(x_1) D_T+G>0\right\}
  \end{equation*}
 with positive probability whence $v(x_2)>v(x_1).$

Establish continuity of $v(x)$ on the interval $(0, \E_{\widetilde{P}}G)$. Let  $x_n\uparrow x>0$ then $$\E_{\widetilde{P}}\left(v(x_n)D_T+G\right)^+ =x_n\uparrow\E_{\widetilde{P}}\left(v(x)D_T+G\right)^+=x.$$
It means that $\E_{\widetilde{P}}\left(v(x_n)D_T+G\right)^+>\frac{x}{2}>0$ for $n>n_0.$ So, $v(x_n)\in(v(\frac{x}{2}),v(x))$ is increasing and bounded and if   $v(x_{n })\uparrow v \neq v(x)$ then
$$\E_{\widetilde{P}}\left(v(x_{n_k})D_T+G\right)^+\uparrow \E_{\widetilde{P}}\left(v D_T+G\right)^+\neq \E_{\widetilde{P}}\left(v(x)D_T+G\right)^+ $$ that is a contradiction. Continuity from the right is proved the same way,  so $v(x)$ is continuous in $x\in(0, \E_{\widetilde{P}} G)$. Concerning the range of the values of $v(x)$, let   $x_k \downarrow 0$.

 Then on one hand $\E_{\widetilde{P}}\left(v(x_k)D_T+G\right)^+ \downarrow 0$  which implies that  $ \left(v(x_k)D_T+G\right)^+ \downarrow 0 $  ${\widetilde{P}}$-a.s. and  $\pr\left\{G>-v(x_k)D_T\right\}\downarrow 0$, or, that is the same, $\pr\left\{v(x_k)D_T\leq -G\right\}\uparrow 1$. On the other hand $v(x_k)\downarrow v\in \mathbb{R}\cup\{-\infty\}$. Therefore $\pr\left\{v \leq \mathrm{ess inf}_{\omega\in \Omega}(-G/D_T)\right\}=1$, and it follows from the continuity of $v(x)$ that   $v=r,$ where $r=-{{\mathrm{ess}\sup}}_{\omega\in \Omega}G.$ Evidently, $\E_{\widetilde{P}}\left(0+G\right)^+=\E_{\widetilde{P}} G$, so,  $v(x)\uparrow 0$ as $x\uparrow \E_{\widetilde{P}} G$ and lemma is proved.
  \end{proof}

Note  that
$$\E\left(\left(G+v(x)D_T\right)^{+}\right)^2\leq 2\E G^2+2 v^2(x)\E D^2_T<+\infty.$$
  Then it follows from the completeness of the market that $\left(G+v(x)D_T\right)^{+}$ admits the representation
\begin{equation}\label{repre1}
\left(G+v(x)D_T\right)^{+}=\E_{\widetilde{P}} \left(G+v(x)D_T\right)^{+} +\int_0^T\widetilde{\xi}_s d\widetilde{S}_s=x+\int_0^T\widetilde{\xi}_s d\widetilde{S}_s
\end{equation}
with some $\widetilde{\xi}\in \Xi$.
\begin{theorem}
  \label{thm3} Let $g\in(0,\E_{\widetilde{P} }G)$ be fixed. Consider $v(g)$ that is the unique solution of equation \eqref{MMM-1-2} with $x=g$ and the  representation \eqref{repre1} of the random variable $\left(G+v(g)D_T\right)^{+}$:
  \begin{equation}\begin{gathered}\label{MMM-3-1}\left(G+v(g)D_T\right)^{+}=g+\int_0^T \widetilde{\xi}_s d\wt{S}_s.\end{gathered}\end{equation}
  Then $\widetilde{\xi}$ is the solution of the minimization problem $[S({g,0;{G}})]$.
  \end{theorem}
  \begin{proof}
On one hand,  $\widetilde{\xi}\in   \Xi(g,0)$. On the other hand, let $\eta\in   \Xi(g,0)$. Then
\begin{align*}
&\E\left(G-g-\int_0^T \eta_s d\wt{S}_s\right)^2\\
&=\E\left(-v(g)D_T+\left(G+v(g)D_T\right)^+-\left(G+v(g)D_T\right)^{-}-g-\int_0^T\eta_s d \wt{S}_s\right)^2\\
&=\E\left(-v(g)D_T-\left(G+v(g)D_T\right)^{-}+\int_0^T(\widetilde{\xi}_s-\eta_s )d \wt{S}_s\right)^2\\
&=\E\left( -v(g)D_T-\left(G+v(g)D_T\right)^{-} \right)^2+\E\left(\int_0^T (\widetilde{\xi}_s-\eta_s ) d\wt{S}_s\right)^2\\
&+2\E\left(-v(g)D_T-\left(G+v(g)D_T\right)^{-} \right)\left(\int_0^T  (\widetilde{\xi}_s-\eta_s ) d\wt{S}_s\right).
\end{align*}
 Note that 
 \begin{align*}
 & -v(g)D_T-\left(G+v(g)D_T\right)^{-}=G - g-\int_0^T \widetilde{\xi}_s d\wt{S}_s;\\
 &\int_0^T  (\widetilde{\xi}_s-\eta_s ) d\wt{S}_s=\left(G+v(g)D_T\right)^{+}-g-\int_0^T\eta_s d\wt{S}_s;\\
 &-\E\left(G+v(g)D_T\right)^{-}\left( \left(G+v(g)D_T\right)^{+}-g-\int_0^T\eta_s d\wt{S}_s\right)\\
 &= \E\left(G+v(g)D_T\right)^{-}\left(g+\int_0^T\eta_s d\wt{S}_s\right)\geq 0;
\end{align*}
$$\E \left(D_T \int_0^T (\widetilde{\xi}_s-\eta_s ) d\wt{S}_s\right)=
 \E_{\widetilde{P}} \int_0^T (\widetilde{\xi}_s-\eta_s ) d\wt{S}_s=0.$$
 Taking this into account, we continue:
\begin{align*}
&\E\left( -v(g)D_T-\left(G+v(g)D_T\right)^{-} \right)^2+\E\left(\int_0^T (\widetilde{\xi}_s-\eta_s ) d\wt{S}_s\right)^2\\
&+2\E\left(-v(g)D_T-\left(G+v(g)D_T\right)^{-} \right)\left(\int_0^T  (\widetilde{\xi}_s-\eta_s ) d\wt{S}_s\right)\\
&=\E\left( G - g-\int_0^T \widetilde{\xi}_s d\wt{S}_s\right)^2+\E\left(\int_0^T (\widetilde{\xi}_s-\eta_s ) d\wt{S}_s\right)^2\\
&-2\E\left(G+v(g)D_T\right)^{-}\left( \left(G+v(g)D_T\right)^{+}-g-\int_0^T\eta_s d\wt{S}_s\right)\\
&\geq \E\left( G -g-\int_0^T \widetilde{\xi}_s d\wt{S}_s\right)^2 +\E\left(\int_0^T (\widetilde{\xi}_s-\eta_s )d\wt{S}_s\right)^2,
\end{align*}
 and the proof follows.
 \end{proof}
 \begin{rem} {\rm \begin{enumerate}
 \item In the case $\widetilde{P}=P,$  the process $\wt{S}$ is a $P$-martingale and $D_T \equiv 1.$
 \item In the case $\widetilde{P}=P,$  the method of solving the mean-variance hedging problem can be described as follows.  We are focused on the case when {$0<x<\E\ H$} and $\E (H| \wt{\Fa}_T) \ge \wt{H}.$  The initial capital $x$ should be divided into two parts $\wt{x}=\E \wt{H}\ge 0$ and $x-\wt{x}\geq 0.$ A hedger with initial capital $\wt{x}$ constructs the strategy $\wt{\eta}$ which replicates the claim $\wt{H}$ using  information $\wt{\mathbb{F}}.$ The process $\wt{\eta}$ and the new initial capital are obtained from {martingale} representation of $\wt{H}.$ Then the hedger should construct the strategy $\wt{\xi}$ such that the terminal wealth $x-\wt{x}+\int_{0}^{T}\wt{\xi}_t d \wt{S}_t$ is non-negative a.s. and $\wt{\xi}$ minimize the mean-square difference of a new claim $\wt{G}:=\E (H|\wt{\Fa}_T)$ and the terminal wealth. The strategy $\wt{\xi}$ is obtained from martingale representation of random variable $(\wt{G}+v(x-\wt{x}))^+,$ where $v(x-\wt{x})$ is such that  $\E (G+v(x-\wt{x}))^+=x-\wt{x}.$ Thus, the resulting strategy is $\wt{\eta}+\wt{\xi}.$ 
 \end{enumerate} }
 \end{rem}
 
\section{ Application to the model with two correlated geometric Brownian motions}
\begin{example}{\rm
\label{ex_S_H}
Let $\wt{W},\widehat{W}$ be two independent Wiener processes under measure $P$ and $W=\rho \wt{W} + \sqrt{1-\rho^2}\widehat{W}.$ Let filtration $\mathbb{F}$ be generated by two-dimensional Wiener process $(\wt{W},\widehat{W})$, filtration $\wt{\mathbb{F}}$ be generated by $\wt{W}.$ Consider the case when an unobservable risk asset is the semimartingale $S=\{S_t=\exp\left(W_t+\int_{0}^{t}a(s)ds\right),t\ge 0 \}$ and an observable risk asset is the semimartingale $\wt{S}=\{\wt{S}_t=\exp(\wt{W}_t+a t), t\ge 0\}$ and the market $\Sigma=\{1,S,\wt{S}\}.$ Let $\{a(s),s\ge 0\}$ be a deterministic function from $L_1[0,T],$ $a$ be some positive constant. Let the contingent claim $G=H(S_T)$ be a square-integrable random variable, where $H:\R\to \R_+$ is real-valued non-decreasing measurable function of polynomial growth at infinity. For initial capital $g\in (0,\E_{\wt{P}} G)$ we find the solution of minimization problem $[S({g,0;{G}})].$

Let $\wt{B}_t:=\wt{W}_t+(a+1/2)t, t\geq 0.$ Then $\wt{S}_t=\exp(\wt{W}_t+a t) = \exp(\wt{B}_t-t/2),t\geq 0.$ If  $\wt{B}$ is a $\wt{P}$-Wiener process, then $\wt{S}$ is a $\wt{P}$-martingale. Denote $a_1:=a+1/2.$ It follows from Girsanov's theorem that {$(\wt{B_t},\widehat{W}_t),t\in [0,T]$} is two-dimensional Wiener process under measure $\wt{P}$ with Radon-Nikodym derivative
\begin{align*}
D_T=\frac{d \wt{P}}{d P} \bigg\vert_{[0,T]} &= \exp\left\{-\int_{0}^T a_1 d \wt{B}_s-\frac12 \int_0^T a_1^2 ds\right\}\\
&= \exp\left\{-a_1\wt{B}_T-\frac{a_1^2 T}{2}\right\} = \exp\left\{-a_1\ln \wt{S}_T+\frac{T}{2}(a^2-\frac{1}{4})\right\}.
\end{align*}
In order to give the explicit solution of minimization problem $[S({g,0;{G}})]$ we make the following steps.

At first, we find $\wt{G}=\E_{\wt{P}}(G|\wt{\Fa}_T).$ We have that
\begin{align*}
 G=H\left(e^{W_T+\int_{0}^{T}a(s)ds}\right)&=H\left(e^{\rho \wt{W}_T+\sqrt{1-\rho^2}\widehat{W}_T+\int_{0}^{T}a(s)ds}\right)\\
 &=H\left(e^{\rho \wt{B}_T+\sqrt{1-\rho^2}\widehat{W}_T+\int_{0}^{T}a(s)ds-a_1 \rho T}\right)\\
 &=H\left((\wt{S}_T)^\rho e^{\sqrt{1-\rho^2}\widehat{W}_T+\int_{0}^{T}a(s)ds-a \rho T}\right).
 \end{align*}
Denote $A_T=\frac{1}{\rho}\int_{0}^{T}a(s)ds-a_1 T.$ Then we get
\begin{align}
\nonumber \wt{G}&=\E(G|\wt{\Fa}_T)\\
\label{G_s_2} &=\int_{\R} H\left(e^{\rho \wt{W}_T+y\sqrt{1-\rho^2}+\rho A_T+\rho a_1 T}\right)\frac{\exp\left(-\frac{y^2}{2 T}\right)}{\sqrt{2\pi T}}dy \\
 &=\int_{\R} H\left(e^u\right)\frac{\exp\left(-\frac{(u-\rho (\wt{W}_T+a_1 T + A_T))^2}{2 T(1-\rho^2)}\right)}{\sqrt{2\pi T(1-\rho^2)}}du.
 \end{align}
Introduce the function
\begin{equation}
\label{df}
f(x)=\int_{\R}H(e^z)\frac{\exp\left(-\frac{(z-\rho x-\rho A_T)^2}{2T(1-\rho^2)}\right)}{\sqrt{2 \pi T (1-\rho^2)}}dz,
\end{equation}
and get that $\widetilde{G}=f\left(\wt{W}_T+a_1 T\right)= f\left(\wt{B}_T\right)=f\left(\ln \wt{S}_T+ T/2\right).$

On the second step, we need to find a solution of equation (\ref{MMM-1-2}). In this case the left hand side of (\ref{MMM-1-2}) equals
\begin{equation*}
\E_{\wt{P}}\left(f(\wt{B}_T)+v(g)\exp\left(-a_1\wt{B}_T+\frac{a_1^2 T}{2}\right)\right)^+.
\end{equation*} Define an auxiliary function $h:\R\to\R_+:$
\begin{equation}
\label{dF}
h(x)=f(x)\exp\left(a_1 x-\frac{a_1^2 T}{2}\right),x\in \R.
\end{equation}
Then $f(x)+v(g)\exp\left(-a_1 x+\frac{a_1^2 T}{2}\right)\ge 0 $ if and only if $ h(x)\geq -v(g).$
 Note that $f$ is non-decreasing function. Indeed, the fact that $H$ is a non-decreasing  implies that for any $x_1<x_2$
\begin{align*}
f(x_1)&=\int_{\R} H(e^{\rho x_1 +\sqrt{T(1-\rho^2)}y+\rho A_T})\frac{\exp(-y^2/2)}{\sqrt{2 \pi}}d y\\
&\leq \int_{\R} H(e^{\rho x_2 +\sqrt{T(1-\rho^2)}y+\rho A_T})\frac{\exp(-y^2/2)}{\sqrt{2 \pi}}d y=f(x_2).
\end{align*}
Moreover,  $h$ is non-decreasing too, and for $h$ we can define a generalized inverse function $h^{(-1)}(x):=\inf\{y: h(y)>x\},x\in \R.$ So, equation (\ref{MMM-1-2}) is rewritten in the following form
\begin{align*}
g&=\E_{\wt{P}}\left(f(\wt{B}_T)+v(g)\exp\left(-a_1\wt{B}_T+\frac{a_1^2 T}{2}\right)\right)^+\\
&=\E_{\wt{P}}\left(f(\wt{B}_T)+v(g)\exp\left(-a_1\wt{B}_T+\frac{a_1^2 T}{2}\right)\right)\mathbbm{1}\{\wt{B}_T\geq h^{(-1)}(-v(g))\}\\
&=\int_{h^{(-1)}(-v(g))}^{+\infty} f(x)\frac{\exp\left(-\frac{x^2}{2T}\right)}{\sqrt{2 \pi T}}dx +v(g)e^{a_1^2 T}\Phi\left(-\frac{h^{(-1)}(-v(g))}{\sqrt{T}}-a_1\right),
\end{align*}
where  $\Phi(x)=\int_{-\infty}^{x}\frac{\exp(-y^2/2)}{\sqrt{2 \pi}}dy,x\in \R$ is the standard normal cumulative distribution function.
The existence of solution $v(g)$ follows from Lemma \ref{v_eq}.
Now, we find the value $\E_{\wt{P}}\wt{G}.$
\begin{equation*}
\E_{\wt{P}}\wt{G}=\E_{\wt{P}}G = \E_{\wt{P}} H\left(e^{\rho \wt{B}_T+\sqrt{1-\rho^2} \widehat{W}_T+\rho A_T}\right)= \int_{\R} H\left(e^{y+\rho A_T}\right)\frac{\exp\left(-\frac{y^2}{2T}\right)}{\sqrt{2 \pi T}}dy.
\end{equation*}
On the next step we need to specify the integral representation of $\left(\wt{G}+v(g)D_T\right)^+.$ Define an auxiliary function \begin{align}
\nonumber \wt{h}(x,\beta)&=(f(x)-\beta \exp\{-a_1x + a_1^2 T/2\})\mathbbm{1}\{x\ge h^{(-1)}(\beta)\}\\
\label{hdef}&=\exp\{-a_1x + a_1^2 T/2\}(h(x)-\beta)\mathbbm{1}\{x\ge h^{(-1)}(\beta)\}, x\in \R, \beta >0.
\end{align}

{In this example we deal with It\`o stochastic integrals and hence we use the Clark-Ocone representation  in Brownian setting  for random variables from the space $\mathbb{D}_{1,2}$ (\cite[Theorem 4.1., Definition 3.1.]{Oksendal_2009}). From \eqref{df} we see that $f\in C^1(\R).$  Then by the chain rule (\cite[Theorem 3.5.]{Oksendal_2009}) we have $\wt{G}=f(\wt{B}_T)\in \mathbb{D}_{1,2}$ and}
$$\widetilde{G} =\E_{\wt{P}}\widetilde{G}+\int_0^T \E_{\wt{P}}(D_t \widetilde{G}| \widetilde{\Fa}_t)d \widetilde{S}_t=\E_{\wt{P}}\widetilde{G}+\int_0^T \frac{1}{\wt{S}_t}\E_{\wt{P}}(D_t \widetilde{G}| \widetilde{\Fa}_t)d \widetilde{B}_t,$$ where $D_t F$ is the stochastic derivative of a random variable $F.$ For the properties of stochastic derivatives we refer to \cite{Oksendal_2009}.
In the Brownian setting we have
$$D_t \widetilde{G} =D_t f(\widetilde{B}_T) =\frac{\partial f (\widetilde{B}_T)}{\partial x}\mathbbm{1}_{[0,T]}(t).$$ 

From (\ref{df}) we obtain
\begin{align}
\nonumber
\frac{\partial f (x)}{\partial x}&=\frac{\partial}{\partial x} \int_{\R}H(e^z)\frac{\exp\left(-\frac{(z-\rho x-\rho A_T)^2}{2T(1-\rho^2)}\right)}{\sqrt{2 \pi T (1-\rho^2)}}dz\\
\nonumber &=\int_{\R}H(e^z)\frac{\exp\left(-\frac{(z-\rho x-\rho A_T)^2}{2T(1-\rho^2)}\right)}{\sqrt{2 \pi T (1-\rho^2)}}\frac{(z- \rho x-\rho A_T)\rho }{T (1-\rho^2)}dz\\
\label{deriv_f} &= \rho \int_{\R}H(e^{\rho x + y\sqrt{T(1-\rho^2)}+\rho A_T})\frac{\exp(-y^2/2)}{\sqrt{2 \pi T(1-\rho^2)}} y  dy.
\end{align}

Then $\left(\wt{G}+v(g)D_T\right)^+=\wt{h}(\wt{B_T},-v(g))$ and $\E_{\wt{P}}\wt{h}(\wt{B_T},-v(g))=g.$  
From Theorem \ref{thm3}, the solution $\wt{\xi}$ of minimization problem $[S({g,0;{G}})]$ is the process from the integral representation of the random variable $\wt{h}(\wt{B}_T,-v(g))=g+\int_{0}^T\wt{\xi}_t d\wt{S}_t=g+\int_{0}^T \wt{\xi}_t (\wt{S}_t)^{-1} d\wt{B}_t.$ We use the Clark-Ocone formula
\begin{equation}
\label{repr_fg_2}
\wt{h}(\wt{B}_T,-v(g))=\E_{\wt{P}} \wt{f}(\wt{B}_T,-v(g)) + \int_{0}^{T} \E [D_t\wt{h}(\wt{B}_T,-v(g))|\wt{\Fa}_t]d \wt{B}_t,
\end{equation}
We see from \eqref{hdef} that the function $\wt{h}$ is not differentiable at $x = h^{(-1)}(-v(g))$, so we cannot use the chain rule directly to
evaluate $D_t\wt{h}(\wt{B}_T,-v(g)).$ However, using the same arguments as in \cite[Example 5.15]{Oksendal_96}, we can approximate $\wt{h}$ by $C^1$ functions $\wt{h}_n$ with the property that $\wt{h}_n(x,-v(g))=\wt{h}(x,-v(g))$ for $|x-h^{(-1)}(-v(g))|\geq\frac{1}{n}.$  
Random variables $\wt{h}_n(\wt{B}_T,-v(g))$ are Wiener polynomials and for all $n>1$ $D_t\wt{h}_n(\wt{B}_T,-v(g))$ exists (see \cite [Lemma A.12]{Oksendal_2009}). 
The space $\mathbb{D}_{1,2}$ consists of all $F\in L^2(P)$ such that there exists Wiener polynomials $F_n$ with the property that $F_n\to F$ in $L^2(P), n\to \infty$ and $\{D_t F_n\}_{n=1}^{\infty}$ is convergent in $L^2(P \times \lambda).$
 Since $\{D_t \wt{h}_n(\wt{B}_T,-v(g))\}_{n=1}^{\infty}$ is convergent and $\wt{h}_n(\wt{B}_T,-v(g))\to \wt{h}(\wt{B}_T,-v(g))$ in $L^2(P),$ we have $\wt{h}(x,-v(g))\in\mathbb{D}_{1,2}.$ Therefore, $D_t \wt{h}(\widetilde{B}_T,-v(g))$ exists. By closability 
of operator $D_t$, we have (see \cite[Theorem A.14]{Oksendal_2009})
\begin{align*}
    D_t\wt{h}(\wt{B}_T,-v(g))&=\lim_{n\to \infty}D_t\wt{h}_n(\wt{B}_T,-v(g))\\
    &=\mathbbm{1}\{\widetilde{B}_T\geq h^{-1}(-v(g))\}D_t \left(f(\widetilde{B}_T)+v(g)e^{-a_1 \wt{B}_T+a_1^2 T/2}\right).
\end{align*}


We know that if $\int_0^T \delta_s d \wt{B}_s=0$ a.s. then $\delta_s=0$ a.s.
Thus, the process $\wt{\xi}$ in the representation of $\wt{h}(\wt{B_T},-v(g))$ is uniquely determined. 


Therefore, the solution $\wt{\xi}_t$ of minimization problem $[S({g,0;{G}})]$ equals
\begin{align*}
\wt{\xi}_t&=\wt{S}_t \E_{\wt{P}}\left[D_t \wt{h}(\wt{B}_T,-v(g)) \big \vert \wt{\Fa}_t\right]\\
&=\wt{S}_t \E_{\wt{P}}\left[\mathbbm{1}\{\wt{B}_T\ge h^{(-1)}(-v(g))\} D_t \left(f(\wt{B}_T)+v(g) \exp\left(-a_1 \wt{B}_T + \frac{a_1^2 T}{2}\right)\right) \bigg \vert \wt{\Fa}_t\right]\\
&=\wt{S}_t \E_{\wt{P}}\left[\mathbbm{1}\{\wt{B}_T\ge h^{(-1)}(-v(g))\} \left(\frac{\partial f(\wt{B}_T)}{\partial x}- v(g) a_1 \exp\left(-a_1 \wt{B}_T + \frac{a_1^2 T}{2}\right)\right) \bigg \vert \wt{\Fa}_t\right]\\
&=\wt{S}_t \E_{\wt{P}}\left[\mathbbm{1}\{\wt{B}_T-\wt{B}_t\ge h^{(-1)}(-v(g))-\wt{B}_t\} \left(\frac{\partial f(\wt{B}_T-\wt{B}_t+\wt{B}_t)}{\partial x}\right) \bigg \vert \wt{\Fa}_t\right]\\
&-v(g)a_1 \wt{S}_t  \exp\left(-a_1\wt{B}_t+\frac{a_1^2 T}{2}\right)\\
&\times \E_{\wt{P}}\left[\mathbbm{1}\{\wt{B}_T-\wt{B}_t\ge h^{(-1)}(-v(g))-\wt{B}_t\} \exp\left(-a_1 (\wt{B}_T-\wt{B}_t) \right) \bigg \vert \wt{\Fa}_t\right].
\end{align*}
Using (\ref{deriv_f}), we  obtain
\begin{align}
\nonumber \wt{\xi}_t&=\wt{S}_t \int_{h^{(-1)}(-v(g))-\wt{B}_t}^{+\infty}\frac{\partial f(x+\wt{B}_t)}{\partial x}\frac{\exp\left(-\frac{x^2}{2(T-t)}\right)}{\sqrt{2\pi (T-t)}}dx\\
\nonumber &-v(g)a_1 \wt{S}_t \exp\left( -a_1 \wt{B}_t+\frac{a_1^2 T}{2}\right)\int_{h^{(-1)}(-v(g))-\wt{B}_t}^{+\infty} \exp\left(-a_1 x \right)\frac{\exp\left(-\frac{x^2}{2(T-t)}\right)}{\sqrt{2\pi (T-t)}}dx\\
\nonumber &=\rho \wt{S}_t \int_{h^{(-1)}(-v(g))-\wt{B}_t}^{+\infty}\int_{\R}H\left(\exp(\rho x + y \sqrt{T(1-\rho^2)}+\rho \wt{B}_t+\rho A_T)\right)\\
\nonumber&\times\frac{ y \exp(- y^2 /2 )}{\sqrt{2\pi T (1-\rho^2)}} \frac{\exp\left(-\frac{x^2}{2(T-t)}\right)}{\sqrt{2\pi (T-t)}}dydx\\
\label{xi_S} &-v(g)a_1 \wt{S}_t \exp\left( -a_1\wt{B}_t+\frac{a_1^2 T}{2}\right)e^{a_1^2 (T-t)/2}\Phi\left(\frac{\wt{B}_t-h^{(-1)}(-v(g))}{\sqrt{T-t}}-a_1\right).
\end{align}
Finally, we give a formula for $\E\left(G-g-\int_0^{T}\wt{\xi}_s d \wt{S}_s\right)^2.$
According to Remark \ref{rem6} the latter value equals $\E(G-\wt{G})^2+\E\left(\wt{G}-g-\int_0^{T}\wt{\xi}_s d \wt{S}_s\right)^2.$ Note that $\E(G-\wt{G})^2=\E G^2- 2 \E (G f(\wt{W}_T+a_1 T))+\E f^2(\wt{W}_T+a_1 T).$  We have
\begin{equation*}
\E G^2=\int_{\R}H^2 \left(e^{u+\rho A_T+\rho a_1 T}\right)\frac{\exp\left(-\frac{u^2}{2T}\right)}{\sqrt{2 \pi T}} du.
\end{equation*}
Then
\begin{align*}
&\E(G-\wt{G})^2
=\int_{\R}H^2 \left(e^{u+\rho A_T +\rho a_1 T}\right)\frac{\exp\left(-\frac{u^2}{2T}\right)}{\sqrt{2 \pi T}} du\\
&-2\int_{\R}\int_{\R}H\left(e^{ \rho x+y\sqrt{1-\rho^2}+\rho A_T + \rho a_1 T}\right)f(x+a_1 T) \frac{\exp\left(-\frac{x^2+y^2}{2T}\right)}{2 \pi T} dy dx\\
&+\int_{\R}f^2(x+a_1 T) \frac{\exp\left(-\frac{x^2}{2T}\right)}{\sqrt{2 \pi T}}dx\\
&=\int_{\R}\left(H^2 \left(e^{u+\rho A_T +\rho a_1 T}\right)-f^2(x+a_1 T)\right)\frac{\exp\left(-\frac{x^2}{2T}\right)}{\sqrt{2 \pi T}}dx.
\end{align*}
It follows from integral representation of $(\wt{G}+v(g)D_T)^+$ that
 \begin{align}
\nonumber&\E\left(\wt{G}-g-\int_0^{T}\wt{\xi}_s d \wt{S}_s\right)^2=\E (\wt{G}-(\wt{G}+v(g)D_T)^+)^2\\
\nonumber&=\E (\wt{G})^2\mathbbm{1}\{\wt{G}+v(g)D_T \leq 0\}+(v(g))^2\E D_T^2\mathbbm{1}\{\wt{G}+v(g)D_T > 0\}\\
\nonumber&=\int_{-\infty}^{h^{(-1)}(-v(g))}f^2(x+a_1 T)\frac{\exp\left(-\frac{x^2}{2T}\right)}{\sqrt{2 \pi T}}dx\\
\nonumber&+(v(g))^2\int_{h^{(-1)}(-v(g))}^{+\infty} \exp (-2a_1 x-a_1^2 T)\frac{\exp\left(-\frac{x^2}{2T}\right)}{\sqrt{2 \pi T}}dx\\
\nonumber &=\int_{-\infty}^{h^{(-1)}(-v(g))}f^2(x+a_1 T)\frac{\exp\left(-\frac{x^2}{2T}\right)}{\sqrt{2 \pi T}}dx\\
\label{delta_G}&+(v(g))^2\exp (a_1^2 T)\Phi\left(-\frac{h^{(-1)}(-v(g))}{\sqrt{T}}-2a_1\right).
 \end{align}}
\end{example}

\begin{example}{\rm
Consider the same problem as in Example \ref{ex_S_H} with  specific function $H(y)=(y-K)^+, y\in\R.$
At first, the function $f$ from (\ref{df}) has the following form
\begin{align}
\label{df2}
\nonumber
f(x)&=\int_{\R}\left(\exp(\rho x + \sqrt{T(1-\rho^2)}y+\rho A_T)-K\right)^+\frac{\exp(-y^2/2)}{\sqrt{2\pi}}dy\\
\nonumber &=\exp\left(\rho x+ T(1-\rho^2)/2+\rho A_T\right)\Phi\left(\frac{\rho x -\ln K+\rho A_T}{\sqrt{T(1-\rho^2)}}+ \sqrt{T(1-\rho^2)}\right)\\
&-K\Phi\left(\frac{\rho x-\ln K+\rho A_T}{\sqrt{T(1-\rho^2)}}\right).
\end{align}
Denote $\alpha=-\ln K+\rho A_T + T(1-\rho^2)/2$ and $\beta=\sqrt{T(1-\rho^2)/2}.$ Then
\begin{align}
\label{df3}
f(x)&=K\exp\left(\rho x+ \alpha\right)\Phi\left(\frac{\rho x+\alpha}{\sqrt{2}\beta}+ \frac{\beta}{\sqrt{2}}\right)-K\Phi\left(\frac{\rho x+\alpha}{\sqrt{2}\beta}- \frac{\beta}{\sqrt{2}}\right).
\end{align}

We evaluate $\frac{\partial f}{\partial x}.$ Using (\ref{df3}) we get
\begin{align}
\nonumber \frac{\partial f(x)}{\partial x}&=\rho K\exp\left(\rho x+ \alpha\right)\Phi\left(\frac{\rho x+\alpha}{\sqrt{2}\beta}+ \frac{\beta}{\sqrt{2}}\right)\\
\label{deriv_f2}&+\frac{\rho K}{\sqrt{2} \beta} \exp\left(\rho x+ \alpha \right)
\varphi\left(\frac{\rho x+\alpha}{\sqrt{2}\beta}+ \frac{\beta}{\sqrt{2}}\right)- \frac{\rho K}{\sqrt{2}\beta} 
\varphi\left(\frac{\rho x+\alpha}{\sqrt{2}\beta}- \frac{\beta}{\sqrt{2}}\right).
\end{align}

Then
\begin{align}
\nonumber
&\E_{\wt{P}}\left(\left.\mathbbm{1}\{\widetilde{B}_T\geq h^{-1}(-v(g))\} \frac{\partial f (\widetilde{B}_T)}{\partial x}\right|\widetilde{\Fa}_t\right)\\
\nonumber =&\E_{\wt{P}}\left(\left.\mathbbm{1}\{\widetilde{B}_T-\widetilde{B}_t\geq h^{-1}(-v(g))-\widetilde{B}_t\} \frac{\partial f (\widetilde{B}_T-\widetilde{B}_t+\widetilde{B}_t)}{\partial x}\right|\widetilde{\Fa}_t\right)\\
\nonumber =&\frac{\rho K}{\sqrt{2}\beta} e^{\rho \wt{B}_t+ \alpha} \int_{h^{-1}(-v(g))-\widetilde{B}_t}^{+\infty}e^{\rho x}\Phi\left(\frac{\rho x+\rho \widetilde{B}_t+ \alpha}{\sqrt{2}\beta}+ \frac{\beta}{\sqrt{2}}\right)\frac{\exp\left(-\frac{x^2}{2(T-t)}\right)}{\sqrt{2 \pi (T-t)}}dx\\
\nonumber +&\frac{\rho K e^{\rho \wt{B}_t+ \alpha}}{2 \pi \sqrt{(T-t)}}  \int_{h^{-1}(-v(g))-\widetilde{B}_t}^{+\infty}\exp\left(\rho x -\frac{1}{2}\left(\frac{\rho x+\rho \widetilde{B}_t +\alpha}{\sqrt{2}\beta}+ \frac{\beta}{\sqrt{2}}\right)^2- \frac{x^2}{2(T-t)}\right)dx\\
-&\frac{\rho K }{2 \pi \sqrt{(T-t)}}  \int_{h^{-1}(-v(g))-\widetilde{B}_t}^{+\infty}\exp\left(-\frac{1}{2}\left(\frac{\rho x+\rho \widetilde{B}_t +\alpha}{\sqrt{2}\beta}- \frac{\beta}{\sqrt{2}}\right)^2- \frac{x^2}{2(T-t)}\right)dx.
\label{l2}
\end{align}
Combining (\ref{xi_S}) and (\ref{l2}) we obtain the solution $\wt{\xi}_t$ of minimization problem $[(S({g,0;{G}})]:$
\begin{align}
\nonumber
\wt{\xi}_t&=
\frac{\rho K \wt{S}_t e^{\rho \wt{B}_t  + \alpha}} {\sqrt{2}\beta}  \int_{h^{-1}(-v(g))-\widetilde{B}_t}^{+\infty}e^{\rho x}\Phi\left(\frac{\rho x+\rho \widetilde{B}_t+ \alpha}{\sqrt{2}\beta}+ \frac{\beta}{\sqrt{2}}\right)\frac{\exp\left(-\frac{x^2}{2(T-t)}\right)}{\sqrt{2 \pi (T-t)}}dx\\
\nonumber +&\frac{\rho K \wt{S}_t e^{\rho \wt{B}_t  + \alpha}}{2 \pi \sqrt{(T-t)}}  \int_{h^{-1}(-v(g))-\widetilde{B}_t}^{+\infty}\exp\left(\rho x -\left(\frac{\rho (x+\widetilde{B}_t) +\alpha}{2\beta}+ \frac{\beta}{2}\right)^2- \frac{x^2}{2(T-t)}\right)dx\\
\nonumber-&\frac{\rho K \wt{S}_t}{2 \pi \sqrt{(T-t)}}  \int_{h^{-1}(-v(g))-\widetilde{B}_t}^{+\infty}\exp\left(-\left(\frac{\rho x+\rho \widetilde{B}_t +\alpha}{2\beta}- \frac{\beta}{2}\right)^2- \frac{x^2}{2(T-t)}\right)dx\\
\label{xi_S_2} &- v(g)a_1 \wt{S}_t e^{-a_1\wt{B}_t-\frac{a_1^2 t}{2} + a_1^2 T }\Phi\left(\frac{\wt{B}_t-h^{(-1)}(-v(g))}{\sqrt{T-t}}-a_1\right).
\end{align}
}

\end{example}

\subsection{Illustrative example}

Consider a numerical example of minimization problem  $[S({g,0;{G}})]$ in the case of $T=2$, $K=1,$ $a=0.5$ and $a(s)\equiv 1,s\geq 0.$ Let $g\in\{g_1,g_2,g_3,g_4\}=\{0.5,1,2,3\},$  and $\rho \in \{0.3,0.5,0.75\}.$

We simulate the trajectories of the Wiener process $\wt{W}_t$ on time interval $[0,2],$ and so we obtain the trajectories of $\wt{B}_t=\wt{W}_t+(a+1/2)t$ and $\wt{S}_t=\exp(\wt{B}_t-t/2).$ We take one of the obtained sample paths of $\wt{S}_t$ and presented it on Figure \ref{xi_graphics_B}.  By formula (\ref{xi_S_2}), for this trajectory we construct the sample paths of $\wt{\xi}_t,$ with $g=3$ and $\rho=0.3,0.5,0.75.$ On Figure \ref{xi_graphics_S}, we present these sample paths, grouped by $\rho.$

Values of solutions of equation (\ref{MMM-1-2}), values of $\E(G-\wt{G})^2$ and $\E_{\wt{P}}\wt{G}$ are presented on Table \ref{KG_S}. We see that $\wt{P}$-average of $G$ decreases when $\rho$ increases.

The values of \eqref{delta_G} are presented on Table \ref{Tab2} and they are naturally decreasing to 0 as $g\uparrow\E_{\wt{P}}\wt{G}.$ However, the total risk is always greater the than $\E(G-\wt{G})^2$. On the second part of Table \ref{Tab2} we present sensitivity of \eqref{delta_G} with respect to $g$ by computing the following values
$\frac{1}{min_i}\frac{ min_{i+1} -min_{i}}{ g_{i+1}-g_{i}},$ where $min_i=\E (\wt{G}-(\wt{G}+v(g_i)D_T)^+)^2,$ $i=1,2,3.$

\begin{table}[ht]
\footnotesize
\caption{}
\label{KG_S}
\begin{center}
\begin{tabular}{|r|r r r r|}
\multicolumn{5}{c}{Solutions of equation (\ref{MMM-1-2})}\\
 \hline
  & g=0.5 &g=1 & g=2 & g=3 \\ \hline
 $\rho$=0.3 &  83.7419 & 41.1694 & 17.2824 & 9.18066 \\
 $\rho$=0.5 &  99.4493 & 40.1427 & 12.4501 & 4.90082 \\
 $\rho$=0.75&  110.058 & 31.6334 & 5.00461 & 0.60940 \\
\hline
\end{tabular}
\begin{tabular}{|r|r r|}
\multicolumn{3}{c}{Values of problem [$S({g,0;{G}})$]}\\
 \hline
  &$\E(G-\wt{G})^2$  &  $\E_{\wt{P}}\wt{G}$  \\ \hline
 $\rho$=0.3 & 2497.04 & 10.0949\\
 $\rho$=0.5 & 2315.28 & 6.50358\\
 $\rho$=0.75& 1738.17 & 3.67076\\
\hline
\end{tabular}
\end{center}
\end{table}
\begin{table}[ht]
\caption{Minimum values of problem  [$S({g,0;{G}}$]}
\label{Tab2}
\footnotesize
\begin{center}
\begin{tabular}{|r|r r r r|}
 \multicolumn{5}{c}{Values of \eqref{delta_G}}\\
  \hline
  & g=0.5 &g=1 & g=2 & g=3 \\ \hline
 $\rho$=0.3 & 361.328 & 306.613 & 225.509 & 165.277 \\
 $\rho$=0.5 & 412.641 & 308.070 & 177.939 &  98.553 \\
 $\rho$=0.75& 506.590 & 291.153 &  92.440 &  15.821 \\
\hline 
\end{tabular}~
\begin{tabular}{|r|r r  r|}
 \multicolumn{4}{c}{Changes of \eqref{delta_G}}\\
  \hline
  & g=0.5 &g=1 & g=2  \\ \hline
 $\rho$=0.3 &  -30.3\% & -26.5\%  & -26,7\% \\
 $\rho$=0.5 & -50.7\% & -42.2\%   &  -44.6\% \\
 $\rho$=0.75&  -85.1\% &  -68.3\% &  -82.9\% \\
\hline 
\end{tabular}
\end{center}
\end{table}

\begin{figure}[h]
 \center{ \includegraphics[width=.8\linewidth]{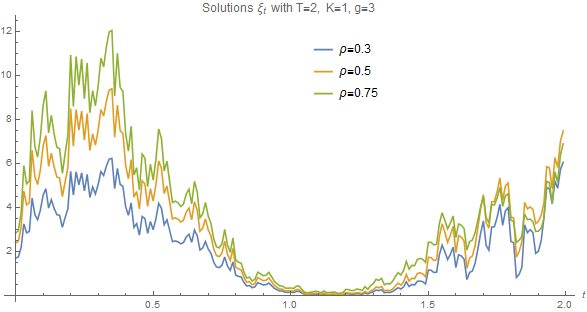}}
  \caption{Sample paths of solutions of minimization problem [$S({g,0;{G}})$]
  \label{xi_graphics_S} with $g=4$}
\end{figure}
~\begin{figure}[h]
 \center{ \includegraphics[width=.8\linewidth]{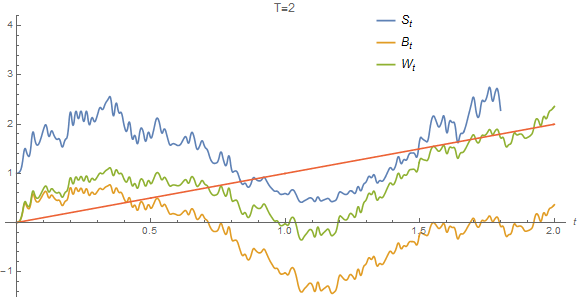}}
  \caption{Sample paths of $\wt{S},\wt{B},\wt{W}$  
  \label{xi_graphics_B}}
\end{figure}


\end{document}